\documentclass[11pt]{article}
\usepackage{amsmath,amsthm,amssymb,cite,enumerate}
\usepackage[a4paper,left=0.7in,right=0.92in]{geometry}
\usepackage[usenames]{color}

\newtheorem{thm}{Theorem}[section] 
\newtheorem{cor}[thm]{Corollary} 
\newtheorem{lem}[thm]{Lemma} 
\newtheorem{prop}[thm]{Proposition}
 
\theoremstyle{definition} 
\newtheorem{defn}[thm]{Definition}
  
\theoremstyle{remark}  
  
\def\beq{\begin{eqnarray}}  
\def\eeq{\end{eqnarray}} 
\def \be {\begin{equation}}
\def \ee {\end{equation}}

\def \bn {\mbox{{\bf n}}}
\def \bm {\mbox{{\bf m}}}

\def \bl {\mbox{\boldmath{$\ell$}}}

\def\bsp{\begin{split}}  
\def\esp{\end{split}}

\def\d{\mathrm{d}}

\def \bk {\mbox{\boldmath{$k$}}}

\newcommand{\dsu}[2]{\nabla\smallsup{#1\!\dot{#2}}}

\newcommand{\smallsup}[1]{^{\scriptscriptstyle #1}}
\newcommand{\smallsub}[1]{_{\scriptscriptstyle #1}}

\newcommand{\sou}[1]{o\smallsup{#1}}
\newcommand{\sod}[1]{o\smallsub{#1}}
\newcommand{\csou}[1]{{\bar o}\smallsup{\dot{#1}}}

\newcommand{\siu}[1]{\iota\smallsup{#1}}
\newcommand{\sid}[1]{\iota\smallsub{\!#1}}
\newcommand{\csiu}[1]{{\bar \iota}\smallsup{\dot{#1}}}

\begin{document}

\title{Universal spacetimes in four dimensions}

\author{S. Hervik$^\diamond$,  V. Pravda$^\star$,  A. Pravdov\' a$^\star$\\
\vspace{0.05cm} \\
{\small $^\diamond$ Faculty of Science and Technology, University of Stavanger}, {\small  N-4036 Stavanger, Norway}  \\
{\small $^\star$ Institute of Mathematics of the Czech Academy of Sciences}, \\ {\small \v Zitn\' a 25, 115 67 Prague 1, Czech Republic} \\
 {\small E-mail: \texttt{sigbjorn.hervik@uis.no, 
pravda@math.cas.cz, pravdova@math.cas.cz}} }

\maketitle


\begin{abstract}
 Universal spacetimes are exact solutions to all higher-order theories of gravity.
We study these spacetimes in four dimensions and {provide necessary and sufficient conditions for universality for all Petrov types except of type II}.  We
show that all universal spacetimes in four dimensions are algebraically  special and Kundt. 
 Petrov type D universal spacetimes are necessarily direct products of two 2-spaces of  constant and equal curvature.
Furthermore, type  II universal spacetimes necessarily possess a null recurrent direction and they  admit the above type D direct product metrics as a limit.  Such spacetimes represent gravitational waves propagating on these backgrounds.
Type III universal spacetimes are also  investigated. We determine necessary and sufficient conditions for universality and present an explicit example of a type III universal Kundt non-recurrent metric.  
 
\end{abstract}

\section{Introduction}

Theories of gravity with the Lagrangian of the form
\be
{ L}={ L}(g_{ab},R_{abcd},\nabla_{a_1}R_{bcde},\dots,\nabla_{a_1\dots a_p}R_{bcde})\label{Lagr}
\ee
are natural geometric generalizations of  Einstein gravity.
Many theories of this form, such as Einstein-Weyl gravity, quadratic gravity, cubic gravity, L(Riemann) gravity and their solutions, have been studied in recent years, often motivated by attempts to understand a quantum description
of the gravitational field (see e.g. \cite{Lu:2015cqa,Hennigar2016,Bueno2017,Gullu2011,Gurses2013,Pravda2017PRD} and references therein).

The complexity of the field equations is in general increasing considerably with each term added to the  Einstein-Hilbert action. Thus, very few exact solutions to generalized theories of gravity are known and naturally,  { to examine various mathematical and physical aspects of  these  theories,  } authors often resort to perturbative or  numerical methods.  

However, there exists  a special class of spacetimes, {\it universal spacetimes},
that simultaneously solve vacuum field equations of  all theories of gravity  
with the Lagrangian of the form
\eqref{Lagr}. {Particular examples of such spacetimes were first discussed in the context of string theory \cite{AmaKli89,HorSte90}  and in the context of spacetimes with vanishing quantum corrections \cite{Coleyetal08}. }
The formal definition of universal metrics  reads \cite{Coleyetal08}
\begin{defn}
	\label{univ}
	A metric is {\it universal} if all conserved symmetric rank-2 tensors constructed  
	from  the  metric, the Riemann tensor and its covariant  derivatives
	of arbitrary order are multiples of the metric.
\end{defn}
{Note that from the conservation of the Einstein tensor, it immediately follows that
universal spacetimes are necessarily Einstein spaces.}

In previous works \cite{HerPraPra14,univII}, we studied  necessary and sufficient conditions for universal
spacetimes in an arbitrary dimension. For example, we have proved  \cite{HerPraPra14}

\begin{prop}
	\label{prop_univCSI}
	{ A universal  spacetime is necessarily a CSI spacetime\footnote{CSI {  (constant scalar curvature invariant)} spacetimes  are spacetimes, for which  all curvature invariants constructed from the metric, the Riemann tensor and its covariant derivatives of arbitrary order are constant, see e.g. \cite{ColHerPel06}.
		}.} 
\end{prop}
Note that CSI is a necessary but not a sufficient condition for universality.

For type  N {(employing the higher-dimensional algebraic  classification of tensors \cite{classletter})}, we have  found necessary and sufficient conditions for universality \cite{HerPraPra14}:

\begin{prop}
	\label{prop_typeN}
	{ A type N spacetime is universal if and only if it is an Einstein Kundt spacetime\footnote{Kundt spacetimes are spacetimes admitting null geodetic conguence with vanishing shear, expansion and twist (see e.g. \cite{Stephanibook,OrtPraPra12rev}).}.}
\end{prop}

For type III \cite{classletter}, we have found sufficient conditions for universality
\cite{HerPraPra14}:
\begin{prop}
	\label{prop_typeIII}
	Type III, $\tau_i=0$ Einstein Kundt spacetimes obeying  
	\be
	C_{acde} C_{b}^{\ cde}=0 \label{QGterm} 
	\ee
	are universal. 
\end{prop}
Note that the $\tau_i=0$ condition implies that the {null} Kundt direction $\bl$ is recurrent\footnote{{Recurrent null vector $\bl $  obeys $ \ell_{a;b} \propto \ell_a \ell_b$.}}
 and that the cosmological constant $\Lambda$ vanishes. Thus, these spacetimes are Ricci-flat.

In \cite{univII}, we have studied type II  and D universal spacetimes. It has turned out that this problem is dimension dependent. For instance, we have proved the non-existence of such spacetimes in five dimensions, while we have provided examples of type D universal spacetimes in any composite  number  dimension as well as  examples of type II universal 
 spacetimes in  various dimensions.

Note that while all known universal spacetimes in dimension $d \geq4 $ \cite{Coleyetal08,HerPraPra14,univII} are algebraically special\footnote{Algebraically special spacetimes are spacetimes of Weyl/Petrov types II, D, III, N, and O.} and Kundt, the existence {algebraically general (type I or  G)} or non-Kundt universal spacetimes  has not been  excluded. 

{Although the results  stated above valid in all dimensions considerably  constrain the space of universal spacetimes by giving various necessary conditions, so far the full 
set of neccessary and sufficient conditions for universality has been known only
for Weyl type  N spacetimes (proposition \ref{prop_typeN}). }

In this work, we focus on the case of four dimensions.  {This leads to a simplification of the problem and in fact it allows us to find necessary and sufficient conditions for universality for all algebraic types except of the type II.}

 In section \ref{sec-typeI}, we prove the {\em non-existence of Petrov type I  universal
spacetimes in four dimensions}. In fact, in combination with further results presented here and in \cite{HerPraPra14}, we find that 
\begin{prop}\label{propalgspec}
Four-dimensional  universal spacetimes are necessarily algebraically special and Kundt.
\end{prop}

Section \ref{secD} is devoted to Petrov type D spacetimes.  
The main result of this section are necessary and sufficient conditions for universality for type D.
\begin{prop}
	\label{propD}
	A four-dimensional type D spacetime is universal if and only if it is a direct product of two 2-spaces of  constant curvature with  the Ricci scalars of the both 2-spaces being equal.
\end{prop}

{Section \ref{secII} focuses on type II universal spacetimes. This is the only  case
for which we do not arrive at a  full set of necessary and sufficient conditions for universality. Nevertheless, we obtain certain necessary conditions. In  particular,}
 we find that these spacetimes 
{\em necessarily admit a recurrent Kundt null direction and that they are Kundt extensions of type D universal backgrounds} discussed above. Thus, they represent gravitational waves propagating on these backgrounds.  Furthermore, we prove that to examine necessary conditions for universality, it is sufficient to consider only rank-2 
tensors linear or quadratic  in $\nabla^{(k)} C$, $k \geq 1$, or rank-2 tensors not containing derivatives of the Weyl tensor.

In section \ref{secIII}, we study type III universal spacetimes and we arrive at necessary and sufficient conditions.
\begin{prop}
	\label{propIII}
	A four-dimensional type III spacetime is universal if and only if it is an Einstein Kundt spacetime  obeying $F_2\equiv C^{pqrs}_{\phantom{pqrs};a} C_{pqrs;b}=0$.
\end{prop}
We also present an explicit type III Kundt Ricci-flat metric 
with $\tau_i \not=0$ and vanishing $F_2$, providing thus an  example of type III non-recurrent universal metric.  

Finally, {in section \ref{concl} we briefly summarize the main results and}  { in table \ref{table1} we compare known necessary/sufficient conditions for  universality for various algebraic types in four and  higher  dimensions. }  We also point out that VSI spacetimes (spacetimes with all scalar curvature invariants vanishing \cite{Pravdaetal02}) are not necessarily universal. 

 Note that all results in the following sections apply to  four dimensions and often this will  not be  stated 
	explicitly.
	We will employ the standard four-dimensional Newman-Penrose formalism  summarized e.g. in \cite{Stephanibook}. Occasionally, to connect with previous higher-dimensional results, we will also refer to  the four-dimensional version of the higher-dimensional real null frame formalism (see e.g. \cite{OrtPraPra12rev} and references therein).

\section{Type I universal spacetimes do not exist}
\label{sec-typeI}

In  this section, we prove the non-existence of type I universal spacetimes. By proposition
\ref{prop_univCSI}, we can restrict ourselves to CSI spacetimes.

It has been shown in \cite{ColHerPel09b} that CSI spacetimes
in four dimensions are either (locally) homogeneous or CSI degenerate Kundt metrics\footnote{Degenerate Kundt spacetimes  \cite{ColHerPel09a}  are Kundt spacetimes  with the Riemann tensor and  its covariant derivatives of arbitrary order  aligned and of type II or more special. For example, all Einstein Kundt spacetimes are degenerate Kundt.}. Degenerate Kundt metrics are algebraically special. Thus, it remains to study type I 
locally homogeneous spacetimes.

Theorem 12.5  of \cite{Stephanibook} and the results given below this theorem imply that ``there are no homogeneous Einstein spaces with $\Lambda \not= 0$ of types I or II''.
Thus for type I universal spacetimes, we have to restrict ourselves to the Ricci-flat case.

Theorem 12.1  of \cite{Stephanibook} states that all non-flat Ricci-flat homogeneous solutions with a multiply transitive group are certain plane waves (of type N).
Theorem 12.2 of \cite{Stephanibook} states that the only vacuum solution admitting a simply transitive $G_4$ as its maximal group of motions is given by
\be
k^2 \d s^2 = \d x^2 + \mathrm{e}^{-2x} \d y^2 + \mathrm{e}^x \left[\cos {\sqrt{3}x} (\d z^2 - \d t^2) - 2 \sin {\sqrt{3}x} \d z \d t        \right] \label{typeICSI}
\ee
with $k$ being an arbitrary constant. 
Thus, this metric 
 is the only type I CSI Einstein metric and  the only type I candidate for a universal metric.

However, it can be shown by a direct calculation that for metric \eqref{typeICSI}, 
a rank-2 tensor $F_2 \equiv C^{pqrs}_{\phantom{pqrs};a} C_{pqrs;b}$ is conserved 
and not proportional to the metric $(F_2)^a_{\ b}=$diag$(0,-48k^2\delta_{ij})$. Thus metric \eqref{typeICSI} is not universal. We conclude with

\begin{lem}
Universal spacetimes in four dimensions are necessarily algebraically special. 
\end{lem}


\section{Type D universal spacetimes}
\label{secD}

Let us proceed with examining type D universal spacetimes.

Without loss of generality, we choose a frame aligned with both multiple principal null directions (PNDs), for which the following frame components of the Weyl tensor vanish
\be
\Psi_0=\Psi_1=\Psi_3=\Psi_4=0.\label{Dkomp}
\ee

The standard complex curvature invariant $I$ (see e.g. \cite{Stephanibook}) can be expressed in terms of the Weyl components as
\be
I=\Psi_0 \Psi_4 -4 \Psi_1 \Psi_3 + 3 {\Psi_2}^2=3 {\Psi_2}^2.
\label{I}
\ee
Thus the CSI condition implies 
\be \Psi_2={\rm const}.
\ee  
Then, the Bianchi equation (7.32e)  of \cite{Stephanibook} for
Einstein spacetimes reduces to
\be
D \Psi_2 + \frac{1}{12} D R = 3 \rho \Psi_2=0.\label{rho0}
\ee
Since $\Psi_2\not= 0$, $\rho$ vanishes and therefore a type D  Einstein CSI
spacetime is necessarily Kundt.

Further Bianchi equations (7.32a)-(7.32h)  of \cite{Stephanibook} imply 
\beq
&& \kappa=0,\ \ \nu=0, \nonumber  \\
 && \sigma=0,\ \   \lambda=0, \nonumber \\
&&\rho=0,\ \ \mu=0, \nonumber \\
&& \tau=0,\ \ \pi=0. \label{typeDspincoeffs}
\eeq

Thus
\begin{lem}
	Type D  Einstein CSI spacetimes are doubly Kundt and both Kundt directions are recurrent.
\end{lem}

Furthermore, taking into account \eqref{typeDspincoeffs}, Ricci identity (7.21h) gives
\be
\Psi_2+\frac{R}{12}=0
\ee
and therefore,
\begin{lem}
	\label{lem_typeDRFCSI}
	Type D Ricci-flat CSI  spacetimes do not exist.
\end{lem}

Let us prove the following lemma: 
\begin{lem}
	\label{}
	Type D  Einstein CSI spacetimes are symmetric (i.e. $R_{abcd;e}=0$).
\end{lem}

\begin{proof}
	This can be more easily shown using spinors, see e.g. \cite{Stephanibook}. The type D Weyl spinor in an adapted frame reads
	\be
	\Psi_{ABCD}= 6 \Psi_2 \sod{(A} \sod{B} \sid{C}\sid{D)}  .
	\ee
	Since due to \eqref{typeDspincoeffs}, the derivatives of the basis spinors satisfy
	\beq
	\dsu{A}{A} \sou{B}&=&\ T\smallsup{A\!\dot{A}}\sou{B},\\
	\dsu{A}{A} \siu{B}&=&- T\smallsup{A\!\dot{A}}\siu{B},
	\eeq
	where
	\be
	T\smallsup{A\!\dot{A}}=\gamma \sou{A}\csou{A}-\alpha \sou{A}\csiu{A}
	-\beta\siu{A}\csou{A}+\epsilon\siu{A}\csiu{A},
	\label{TAA}
	\ee
	and
	\be
	\dsu{A}{A}\Psi_2=0,
	\ee
	we get
	\be
	\dsu{A}{A}\Psi_{ABCD}=0.
	\ee
	
	Thus, 
	\be
	\nabla_e C_{abcd}=\nabla_e R_{abcd}=0
		\ee
	and these spacetimes are symmetric. 
\end{proof}

In four dimensions, type D symmetric spaces are necessarily direct products of two 2-spaces of constant curvature (see chapter 35.2 of \cite{Stephanibook}).
Such a product space is Einstein if and only if the Ricci scalars of both spaces are equal.
It has been shown in \cite{univII} that such direct product spaces are universal.
{This concludes the proof of proposition \ref{propD}.}


\section{Type II universal spacetimes}
\label{secII}

In this section, let us study type II universal spacetimes.

We choose a frame with
\be
\Psi_0 = 0 = \Psi_1.
\ee
Then the curvature invariant $I$ is  given by 
\eqref{I} as in type D
and thus the CSI condition again implies 
\be
\Psi_2={\rm const}. \label{Phi2const}
\ee
For type II Einstein spacetimes, the Goldberg-Sachs theorem implies 
\be 
\kappa=0=\sigma
\ee 
and eq. (7.32e)  of \cite{Stephanibook} again reduces to
\eqref{rho0}
and thus a type II Einstein CSI spacetime is Kundt.
Then, eq. (7.32h)  of \cite{Stephanibook} 
reduces  to
\be
\tau \Psi_2=0
\ee
and therefore 
\begin{prop}\label{lem_K}
Genuine\footnote{Meaning $\Psi_2 \not=0$.} type II Einstein CSI  spacetimes are degenerate  Kundt with a recurrent principal null direction.
\end{prop}

Now, let us study behaviour of the covariant derivatives of the Weyl 
tensor.

\begin{lem}
	For a type II Einstein CSI Kundt spacetime with a recurrent principal null direction, boost
	order of the first covariant derivative of the Weyl tensor is at most $-1$.
\end{lem}

\begin{proof}
	This can be more easily shown using spinors. The type II Weyl spinor in an adapted frame reads
	\be
	\Psi_{ABCD}= 6 \Psi_2 \sod{(A} \sod{B} \sid{C}\sid{D)}
	- 4 \Psi_3 \sod{(A} \sod{B} \sod{C}\sid{D)}
	+  \Psi_4 \sod{A} \sod{B} \sod{C}\sod{D}.\label{Weylspinor}
	\ee
	
	We choose an affinely parametrized Kundt congruence $\bk$ and a 
	frame parallelly propagated along $\bk$
	\be
	0=\kappa=\sigma=\rho=\tau=\epsilon=\pi.
	\ee
	Then,  the derivatives of the basis spinors read
	\beq
	\dsu{A}{A} \sou{B}&=&\ T\smallsup{A\!\dot{A}}\sou{B},	\label{der_spin_o}\\
	\dsu{A}{A} \siu{B}&=&- T\smallsup{A\!\dot{A}}\siu{B}
	+(-\lambda \sou{A}\csiu{A}-\mu \siu{A}\csou{A}+\nu\sou{A}\csou{A})\sou{B},
	\label{der_spin_i}
	\eeq
	where
	\be
	T\smallsup{A\!\dot{A}}=\gamma \sou{A}\csou{A}-\alpha \sou{A}\csiu{A}
	-\beta\siu{A}\csou{A}. 
	\ee
	Note that the covariant derivative $\dsu{A}{A}$ does not increase the boost order
	of the frame spinors $\sou{A}$ and $\siu{A}$. 
	 The Bianchi identity (7.32g) from \cite{Stephanibook} reduces to
	 \be
	 D\Psi_3=0.\label{Bianchi_DPsi3}
	 \ee
	 Thus, taking into account
	\be
	\dsu{A}{A} =\siu{A}\csiu{A}D
	+\sou{A}\csou{A}\Delta
	-\siu{A}\csou{A}\delta
	-\sou{A}\csiu{A}{\bar{\delta}},\label{spinorder}
	\ee
	it follows that
		\be
	\dsu{A}{A}(- 4 \Psi_3 \sod{(A} \sod{B} \sod{C}\sid{D)}
	+  \Psi_4 \sod{A} \sod{B} \sod{C}\sod{D})
	\ee
	contains only b.w. negative terms.

	However, 
	since $\dsu{A}{A}\Psi_2=0$,
	it follows from  \eqref{der_spin_o} and \eqref{der_spin_i}	that
	\be
	\dsu{A}{A} (6 \Psi_2 \sod{(A} \sod{B} \sid{C}\sid{D)})	
	\ee
	also contains only b.w. negative terms, cf. also \eqref{derWeylspinor}.
\end{proof}

\begin{lem}\label{lem_balanII}
	For a type II Einstein CSI Kundt spacetime with a recurrent 
	principal null direction,
	 boost
	order of an arbitrary covariant derivative of the Weyl tensor is at most $-1$.
\end{lem}

\begin{proof} 
	The Ricci equations (7.21d,e,f,g,h,i,q) from \cite{Stephanibook}, using also
	 \eqref{Bianchi_DPsi3} and \eqref{Phi2const},  imply
	\beq
	D\alpha&=&0,\label{Da}\\
	D\beta&=&0,\\
	D\gamma&=&\Psi_2-\frac{R}{24}\ \ \rightarrow\ \ D^2\gamma=0,\\
	D\lambda&=&0,\\
	D\mu&=&\Psi_2+\frac{R}{12},\label{Dmu}\\
	D\nu&=&\Psi_3\ \ \rightarrow\ \ D^2\nu=0,\\
	0&=&\Psi_2+\frac{R}{12},\label{Psi2R}
	\eeq
	respectively. Eq. \eqref{Psi2R} implies that  eq. \eqref{Dmu}
	reduces to
	\be
	D\mu=0
	\ee
	and that $\Psi_2$ is real
	\be
	\Psi_2=-\frac{R}{12}.
	\ee
	
	From the Bianchi equation 
	(7.32c) in \cite{Stephanibook}, it follows
	\be 
	D\Psi_4={\bar{\delta}}\Psi_3+2\alpha\Psi_3-3\lambda\Psi_2.
	\label{Bianchi_DPsi4} 
	\ee 
	Applying the operator $D$ on \eqref{Bianchi_DPsi4} and using the commutator
	\be
	D\delta-\delta D=-({\bar{\alpha}}+\beta)D,\label{comut_Ddelta}
	\ee
	we arrive at
	\be 
	D^2\Psi_4=0.\label{D2Psi4}
	\ee 
	
	Applying the covariant derivative \eqref{spinorder} on the Weyl spinor \eqref{Weylspinor}, we obtain
		\beq
	\dsu{E}{E}\Psi_{ABCD}\!\!\!&=&\!\!\!
	4\sod{(A} \sod{B} \sod{C}\sid{D)}
	[\sou{E}\csiu{E} ({\bar\delta}{\Psi_3}-3\lambda\Psi_2+2\alpha\Psi_3) 
	\nonumber\\
	&&\!\!\!
		+\siu{E}\csou{E}(\delta\Psi_3-3\mu\Psi_2+2\beta\Psi_3)
	+\sou{E}\csou{E}(-\Delta\Psi_3+3\nu\Psi_2-2\gamma\Psi_3)]
	\nonumber\\
	&&\!\!\!
	+\sod{A} \sod{B} \sod{C}\sod{D}
	[\siu{E}\csiu{E}D\Psi_4+\sou{E}\csiu{E}(-{\bar{\delta}}\Psi_4+4\lambda\Psi_3-4\alpha\Psi_4) 
	\nonumber\\
	&&\!\!\!
		+\siu{E}\csou{E}(-\delta\Psi_4+4\mu\Psi_3-4\beta\Psi_4)
		\nonumber\\
		&&\!\!\!
	+\sou{E}\csou{E}(\Delta\Psi_4-4\nu\Psi_3+4\gamma\Psi_4)]
	.\label{derWeylspinor}
	\eeq
	
	Let us employ the balanced scalar/tensor approach in a parallelly propagated frame introduced in \cite{Pravdaetal02}. A scalar $\eta$ with a b.w. b under a constant boost is a balanced
	scalar if $D^{-{\rm b}}\eta=0$ for b$<0$ and $\eta=0$
	for b$\geq 0$. A tensor, whose components are all balanced scalars,
	is a balanced tensor. Obviously, balanced tensors have only b.w. 
	negative components.
	
	While in our case, the Weyl tensor itself is not balanced, we will show that its first derivative \eqref{derWeylspinor} is balanced.
	
	For the first derivative of \eqref{derWeylspinor} to be balanced, we have to show that $D^{\rm b}$
	on a component of b.w. b vanishes. B.w. $-1$,  $-2$,
	and  $-3$   components of \eqref{derWeylspinor}
	read 
	\beq 
	&& \delta\Psi_3,\ {\bar\delta}{\Psi_3},\ 	D\Psi_4,\ \lambda\Psi_2, \ 
	 \ \mu\Psi_2,\  \alpha\Psi_3 ,\
	\beta\Psi_3,\ \\ 
 && 
\Delta\Psi_3, \ \delta\Psi_4,\ 	{\bar{\delta}}\Psi_4,\ \nu\Psi_2,\ \gamma\Psi_3,\ 
\lambda\Psi_3,\ \mu\Psi_3, \ \alpha\Psi_4,\ 
	\beta\Psi_4,
	\\ &&
	\Delta\Psi_4, \ \nu\Psi_3,\ \gamma\Psi_4,
	\eeq 
	respectively.

	Using the Bianchi and Ricci equations and commutators \eqref{comut_Ddelta}
	and
	\be 
	\Delta D-D\Delta =(\gamma +\bar{\gamma})D,\label{comut_DDelta}
	\ee 
	we arrive at
	\beq 
	0&=& 	D({\bar\delta}{\Psi_3})= D(\lambda\Psi_2) 
    =D(	\alpha\Psi_3)
	=D(\delta\Psi_3) =D( \mu\Psi_2) =D( \beta\Psi_3) 
	=D^2 \Psi_4,\nonumber\\
	0&=&D(\Delta\Psi_3) =D^2( \nu\Psi_2)=D^2(\gamma\Psi_3)
	=D^2({\bar{\delta}}\Psi_4) =D (\lambda\Psi_3)\nonumber\\
	 &=&D^2 (\alpha\Psi_4)=D^2 
	(\delta\Psi_4) =D( \mu\Psi_3) =D^2( \beta\Psi_4),\nonumber\\
	0&=&D^3(\Delta\Psi_4)=D^2( \nu\Psi_3) =D^3 (\gamma\Psi_4).
	\eeq 
	This implies that the first derivative of the Weyl tensor is balanced\footnote{Note that the term $\gamma\Psi_2$ that is not balanced
	does not appear in \eqref{derWeylspinor}.}. 
	
	In fact, a covariant derivative of a balanced tensor in a degenerate Kundt spacetime is again
	a balanced tensor (see Lemma B.3 of  \cite{VSIpforms}) and thus   all derivatives of the Weyl tensor are balanced. This concludes the proof.
	
	\end{proof}

As a consequence of lemma \ref{lem_balanII}, 
all tensors of the form $\underbrace{\nabla^{(k_1)}C\otimes\dots\otimes \nabla^{(k_p)}C}_{p\ times}$, $k_i>0$, have boost order $\leq -p$. Since 
a rank-2 tensor has in general boost order $\geq   -2$, all rank-2 tensors
constructed from the Riemann tensor and its covariant derivatives
containing more than two terms of the form $\nabla^{(k)}C$, $k> 0$,  vanish.

{Therefore,  } further necessary conditions for universality may 
follow only from rank-2 tensors linear or quadratic
in $\nabla^{(k)}C$, $k> 0$, or from terms not containing
derivatives of the Weyl tensor.
Now, let us study some of these rank-2 tensors.

{All rank-2 order-4 tensors constructed from the Riemann tensor and its derivatives can be expanded on the FKWC basis \cite{FKWC,Decanini2008}  of rank-2 order-4.}
For Einstein spacetimes, the FKWC basis  of rank-2 order-4 tensors  without derivatives
of the Weyl tensor
 reduces to the four-dimensional identity
\be 
C_{aefg}{C_b}^{efg}=\frac{1}{4}g_{ab} C_{efgh}C^{efgh}, \label{4did}\\
\ee 
while the FKWC basis of rank-2 order-6 tensors  without derivatives
reduces to
\be
R^{pqrs}R_{pqta}{{R_{rs}}^{t}}_{b}, \ \ 
R^{prqs}{R^t}_{pqa}R_{trsb},\ \ 
{R^{pqr}}_s R_{pqrt}{{{{R^s}}_a}^t}_b.
\ee
It turns out that in our case, all these tensors  are either zero or proportional to the metric and thus they do not yield any further
necessary conditions for universality.

A lengthy but straightforward computation  
of the FKWC basis of rank-2, order-6 Weyl polynomials
containing derivatives of the Weyl tensor 
 \cite{FKWC} 
\be
F_1 \equiv C^{pqrs} C_{pqrs;ab}, \ \ \ F_2 \equiv C^{pqrs}_{\ \ \ \ \ ;a} C_{pqrs;b}, 
\ \ \ F_3 \equiv C^{pqr}_{\ \ \ \ a;s} C_{pqrb}^{\ \ \ \ \ ;s} \label{eqF}
\ee
in the Newman-Penrose formalism gives 
\be
F_2=0=F_3. 
\ee
For CSI spacetimes, by differentiating  the identity \eqref{4did}
twice,  we  obtain
\be
F_1+F_2=0\label{F1F2}
\ee
and thus vanishing of $F_2$ implies vanishing of $F_1$.
Thus, all rank-2 order-6 tensors in the FKWC basis either vanish or
are proportional to the metric and give no further necessary conditions for universality.

Explicit examples of type II spacetimes in the context of universality 
were studied in \cite{Coleyetal08} and \cite{univII}. It has been found that
further necessary conditions follow from rank-2 tensors involving higher derivatives of the Weyl tensor, for instance, e.g., from the rank-2 tensor
\be
{R^{cg}}_{eh}{R^{dh}}_{fg}\nabla^{(e}\nabla^{f)}C_{acbd} . \label{GT-RRnablaC}
\ee
Thus, the necessary conditions for universality of type II spacetimes,
Einstein, CSI, Kundt, and recurrent, clearly are not  sufficient.
To find the full set of necessary conditions for type II
at the general level is beyond the scope of this paper.

\subsection{Case $D\Psi_4=0$}

Note that using the Bianchi equations \eqref{Bianchi_DPsi4}, and (7.32d)
and (7.32f) in\cite{Stephanibook}
\beq
\Delta\Psi_3-\delta\Psi_4&=&4\beta\Psi_4 -2(2\mu+\gamma)\Psi_3
+3\nu\Psi_2,\\
-\delta\Psi_3&=&2\beta\Psi_3-3\mu\Psi_2,
\eeq 
respectively,
the first derivative of the Weyl spinor simplifies to
\beq
\dsu{E}{E}\Psi_{ABCD}\!\!\!&=&\!\!\!
\underbrace{D\Psi_4 [4\sod{(A} \sod{B} \sod{C}\sid{D)}
\sou{E}\csiu{E} +\sod{A} \sod{B} \sod{C}\sod{D}
\siu{E}\csiu{E}]}_{b.w. -1}\nonumber\\ &&
+\underbrace{(-\delta\Psi_4+4\mu\Psi_3-4\beta\Psi_4)
[\sod{A} \sod{B} \sod{C}\sod{D}
\siu{E}\csou{E} +4\sod{(A} \sod{B} \sod{C}\sid{D)} \sou{E}\csou{E}]}_{b.w. -2}\nonumber\\ &&
+\underbrace{(-{\bar{\delta}}\Psi_4+4\lambda\Psi_3-4\alpha\Psi_4) \sod{A} \sod{B} \sod{C}\sod{D}\sou{E}\csiu{E} }_{b.w. -2}
\nonumber\\
&&\!\!\!
+\underbrace{(\Delta\Psi_4-4\nu\Psi_3+4\gamma\Psi_4) \sod{A} \sod{B} \sod{C}\sod{D}\sou{E}\csou{E}}_{b.w. -3}
,\label{derWeylspinorsimpl}
\eeq
where
\be 
D(-\delta\Psi_4+4\mu\Psi_3-4\beta\Psi_4)=0.\label{d1}
\ee 
Thus, there is a special subcase of  type II CSI Einstein Kundt spacetimes characterized by $D\Psi_4=0$, for which the first derivative of the Weyl tensor \eqref{derWeylspinorsimpl} contains only b.w.$\leq -2$ terms. Furthermore, 
\beq
D(-{\bar{\delta}}\Psi_4+4\lambda\Psi_3-4\alpha\Psi_4)&=&0,\label{d2}\\
D^2(\Delta\Psi_4-4\nu\Psi_3+4\gamma\Psi_4)&=&D(-4\Psi_3^2+6\Psi_2\Psi_4)=0.\label{d3}
\eeq
Thus, in this case, the first derivative of the Weyl tensor is 
1-balanced\footnote{A scalar $\eta$ with a b.w. b under a constant boost is $1$-balanced
	 if $D^{-{\rm b}-1}\eta=0$ for b$<-1$ and $\eta=0$
	for b$\geq -1$. A tensor, whose components are all $1$-balanced scalars,
	is a $1$-balanced tensor. Obviously, $1$-balanced tensors have only  components of b.w. $\leq -2$ .}.
Using \eqref{Bianchi_DPsi3},  \eqref{Da}--\eqref{D2Psi4}, \eqref{comut_DDelta}, and \eqref{d1}--\eqref{d3} 
the same proof as in sec. 4 of \cite{HerPraPra14} or in sec. 7.1.
of \cite{univII} applies to our case and thus a {\it covariant derivative of
a 1-balanced tensor 
 is 1-balanced}. Therefore, all covariant derivatives of the Weyl tensor are 1-balanced and hence they contain only b.w.$\leq -2$
components. This implies that while studying universality within this class, it is sufficient to  study only  rank-2 tensors linear in derivatives of the Weyl tensor.

Note that the Khlebnikov-Ghanam-Thompson metric discussed
in the context of universality in \cite{Coleyetal08} and \cite{univII} in four and higher dimensions, respectively,  are explicit examples of spacetimes belonging to the 
$D\Psi_4=0$ class.



\subsection{Seed metric for type II universal spacetimes}
\label{sec_seedII}

All type II Einstein recurrent Kundt spacetimes have the metric of the form  
\be 
\d s^2=2\d u (\d v+H\d u+W_x\d x+W_y\d y)+h^2(u,x,y)(\d x^2+\d y^2),\label{dsKundt}
\ee 
where $H=v^2 \Lambda/2 
+vH^{(1)}(u,x,y)+H^{(0)}(u,x,y)$, and $W_i= 
W^{(0)}_i(u,x,y)$. 
The CSI condition implies further that 
$h$ does not depend on $u$.

Consider the one-parameter group of diffeomorphisms of the metric \eqref{dsKundt} defined by $\phi_\lambda: (u,v)\mapsto (ue^{-\lambda},ve^{\lambda})$. This map gives a rescaling of the functions as follows: 
\beq
&&\left(H^{(1)}(u,x,y), H^{(0)}(u,x,y)\right)\longmapsto 
  \left(e^{-\lambda}H^{(1)}(ue^{-\lambda},x,y), e^{-2\lambda}H^{(0)}(ue^{-\lambda},x,y)\right), \\
&& W^{(0)}_i(u,x,y)\longmapsto e^{-\lambda}W^{(0)}(ue^{-\lambda},x,y).
\eeq
This map is a diffeomorphism and leaves the invariants invariant and is the Lorentizan version of the limiting map in \cite{HHY}. Let $p$ be the fixed point of $\phi_\lambda$ given by $(u,v,x^i)=(0,0,x^i)$. Then note that the map $d\phi_\lambda$ induces a boost on the tangent space $T_pM$ which aligns with the natural null-frame of (\ref{dsKundt}). Hence, given an arbitrary curvature tensor $R$ of  (\ref{dsKundt}) with boost weight decomposition $R=\sum_{b\leq 0}(R)_b$, then at $p$
\[ \phi_\lambda^*R=\sum_{b\leq 0}e^{b\lambda}(R)_b=(R)_0+e^{-\lambda}(R)_{-1}+e^{-2\lambda}(R)_{-2}+...\] 
Consequently, 
\[ \lim_{\lambda\rightarrow\infty}\phi_\lambda^*R=(R)_0.\] 
However, in the limit $\lambda \rightarrow \infty$, the metric is a type D metric with the same invariants as the type II metric. We also note that the universality requirement is invariant under this diffeomorphism\footnote{This follows from the fact that $T_{ab}=k g_{ab}$ and $\phi^*_{\lambda}g_{ab}=g_{ab}$ at $p$.}, as well as in its limit, and hence, in this limit,  the metric turns into a universal type D metric having identical invariants. This implies that the ``background'' metric for universal type II metrics are universal type D metrics.

Example \cite{PodOrt03}
\be
\d s^2 = \d s^2_b 
+ [ f(\zeta,u) + {\bar f}({\bar\zeta}, u)] \d u^2 ,
\ee
where $\d s^2_b$
\be
\d s^2_b=\frac{2\d\zeta\d {\bar\zeta}}{\left(1+\frac{1}{2}\Lambda\zeta{\bar\zeta}\right)^2}+2\d u\d v+\Lambda v^2\d u^2,
\ee
is the metric of the (anti-)Nariai vacuum universe with 
$\Lambda > 0$ ($\Lambda  < 0$), and $f(\zeta,u)$
is an arbitrary holomorphic (in $\zeta$)  function 
characterizing the profile of the gravitational
wave. This metric is a special case of metrics considered in \cite{univII} with
\be
H= [ f(\zeta,u) + {\bar f}({\bar\zeta}, u)] 
\ee
and it was conjectured there that such a metric is universal if 
\be
 (\Box^{(1)})^P  H=0, \label{GT-sumBoxP}
\ee
where $P=1, 2$ (note that $\Box^{(0)} H=0$ identically and that
 the vacuum Einstein equations with the cosmological constant $\Box^{(1)} H=0$ imply  $(\Box^{(1)})^2 H=0$). 


\section{Type III universal spacetimes}
\label{secIII}

It follows from the results of Sec. 5.2. of \cite{HerPraPra14} that type III universal spacetimes in four dimensions are Kundt.

In four dimensions for type III, the following identity holds
\be
C_{acde} C_{b}^{\ cde}=0 .\label{QGterm} 
\ee
As a consequence of \eqref{QGterm}, theorem 1.4. of \cite{HerPraPra14} reduces to
\begin{prop}
\label{prop_typeIII}
Type III, recurrent ($\tau_i=0$) Einstein Kundt spacetimes     are universal. 
\end{prop}

Note that it follows directly from the Ricci identity (7.21q) of \cite{Stephanibook} that $\tau=0$ implies that Ricci scalar vanishes and thus these spacetimes are in fact Ricci-flat,  as observed in \cite{HerPraPra14}. 
An explicit example of such a metric is given in \cite{HerPraPra14}.

 Thus, in this section we focus on the non-recurrent ($\tau_i\not=0$) case
 which also allows for $\Lambda\not=0$.

Let us start with proposition 5.1. of \cite{HerPraPra14}
\begin{prop}
\label{prop-balanced}
For type III Einstein  Kundt spacetimes, the boost order of  $\nabla^{(k)} C$  (a covariant derivative  of an arbitrary order of the Weyl tensor)
{  {with respect to the multiple WAND}}  is at most $-1$. 
\end{prop}

A straightforward consequence of the above proposition is 
a generalization of Lemma 5.2. of \cite{HerPraPra14} to Einstein spacetimes:
\begin{lem}
\label{lemma_quadraticC}
For  type III Einstein Kundt spacetimes, a non-vanishing rank-2  
tensor constructed from the metric, the Weyl tensor and its covariant 
derivatives of arbitrary order is at most quadratic in the Weyl tensor
 and its covariant derivatives.
\end{lem}
It has been shown  in the proof of proposition 5.1. of \cite{HerPraPra14}  that 
for type III Einstein  Kundt spacetimes, the Weyl tensor and its covariant derivatives of arbitrary order are  balanced.  Thus it follows: 
\begin{cor}
\label{typeIIIcons}
For type III Einstein  Kundt spacetimes, all rank-2 tensors constructed 
from the Weyl tensor and its covariant derivatives of arbitrary order 
quadratic in the Weyl tensor and its covariant derivatives are conserved.
\end{cor}

In the following, we will employ the formula for
the commutator for an arbitrary tensor:
\be
[\nabla_a,\nabla_b]T_{c_1.\dots c_k}=T_{d\dots c_k}{R^d}_{c_1 ab}+\dots 
+T_{c_1\dots d}{R^d}_{c_kab} \label{commut_der}.
\ee

\subsection{The Ricci-flat case}
\label{sec_III_Rflat}

In the Ricci-flat case,  covariant derivatives in a rank-2 tensor quadratic in the Weyl tensor and its derivatives
effectively commute thanks to lemma \ref{lemma_quadraticC} and \eqref{commut_der}. 
Thus, using the Bianchi identities, one can generalize lemmas 5.3. and 5.4. of \cite{HerPraPra14} to the $\tau_i\not=0$ case
\begin{lem}
\label{lemma_insum}
For  type III Ricci-flat Kundt spacetimes, a rank-2  
tensor constructed from the metric, the Weyl tensor and its
 covariant derivatives of arbitrary order quadratic in $\nabla^{(k)} C$,  $k\geq 0$, 
vanishes if it contains a summation within $\nabla^{(k)} C$. 
\end{lem}

\begin{lem}
\label{lemma_der}
For  type III, Ricci-flat Kundt spacetimes, let us assume that 
a certain rank-2 polynomial quadratic in $\nabla^{(k)} C$ vanishes. 
Symbolically we will write $C^{(1)} C^{(2)} =0$. Then also 
$C^{(1)}_{\ ;f} C^{(2)\ \! ;f} 
=0$.
\end{lem}

First, let us examine conserved rank-2 tensors quadratic in the Weyl tensor from
 the FKWC basis \cite{FKWC} of rank-2, order-6 Weyl polynomials  
 \eqref{eqF}.
In our case, $F_3$  vanishes identically as a consequence of \eqref{QGterm} 
and lemma \ref{lemma_der}. 

On the other hand, $F_2$ is in general non-vanishing (see sec. \ref{sec_F2van}), however, in this case, $F_2=0$ is a necessary condition for universality and  will be assumed in the rest of this section.
From \eqref{F1F2}, vanishing of $F_2$ implies vanishing of $F_1$.

For spacetimes satisfying $F_2=0$, 
the FKWC basis of rank-2, order-6 tensors  vanishes and thus also all rank-2, order-6 Weyl polynomials.

Now, let us  prove universality in  the Ricci-flat case.

\begin{prop}
Type III, Ricci-flat Kundt spacetimes, obeying $F_2=0$ are universal.
\end{prop}

\begin{proof}

By lemma \ref{lemma_quadraticC}, we can limit ourselves to the discussion 
of rank-2 tensors which are linear or quadratic in $\nabla^{(k)} C$, where $k=0,1,\dots$. We start with the quadratic case.

The key tools in the proof are Lemmas \ref{lemma_insum} and \ref{lemma_der}
and the observation that covariant derivatives in a rank-2 tensor quadratic in the Weyl tensor and its derivatives effectively commute.

First, consider rank-2 tensors quadratic in the Weyl tensor and its derivatives with both free indices appearing in the first term $\nabla^{(k)} C$.
Symbolically, such tensors will be written as
$$
C_{ab..;\dots  }C^{....;\dots}, \ C_{a.b.;\dots  }C^{....;\dots}, \
C_{a...;b\dots  }C^{....;\dots}, \ C_{....;ab\dots  }C^{....;\dots},
$$
etc., where $a,b$ are free indices and the dots represent various combinations of dummy indices. We understand that covariant derivatives are of arbitrary high order. 

Using symmetries of the Weyl tensor, the Bianchi identities, by lemma \ref{lemma_insum}, and the fact that here covariant derivatives commute, all above rank-2 tensors can be reduced to 
\be
C_{a.b.;\dots  }C^{....;\dots} = \nabla^{(n)} C_{a.b.} \nabla^{(n-2)} C^{....}\,.
\label{typeIIIproof}
\ee
All indices in  $\nabla^{(n)}$ are dummy indices and by lemma \ref{lemma_insum},
to obtain a non-zero result, they should be contracted with the dummy indices in the second term $\nabla^{(n-2)} C_{....}$.  Due to the symmetries of the Weyl tensor, only two of them can be contracted with  $C_{....}$, while remaining indices are contracted with those of $\nabla^{(n-2)}$. Now by lemma \ref{lemma_der}, the tensor \eqref{typeIIIproof} vanishes since
$$
\nabla^{(2)} C_{a.b.}  C^{....} =0,
$$
as a consequence of vanishing of the  rank-2, order-6 Weyl FKWC basis.

Next, consider rank-2 tensors quadratic in the Weyl tensor and its derivatives with the free indices appearing in both terms. Such tensors reduce
to
\be
C_{a...;\dots  }C^{b...;\dots} = \nabla^{(n)} C_{a...} \nabla^{(n)} C^{b...}.
\ee
In order to get a non-zero result, at most two dummy indices in  $\nabla^{(n)}$
in the first term can be contracted with $C^{b...}$ in the second term. Thus 
$n-2$ indices will appear in both $ \nabla^{(n)}$ terms. By lemma \ref{lemma_der}, the problem thus reduces to determining whether
\be
C_{a...;..  }C^{b...;..} = \nabla^{(k)} C_{a...} \nabla^{(k)} C^{b...},\ \ \ \ k\leq 2,
\ee
vanishes. Cases $k=0,1$ are trivial.
For $k=2$, to obtain a non-trivial result, the indices in the first $ \nabla^{(2)}$ have to be contracted with $ C^{b...}$ and similarly with the second $\nabla^{(2)}$. Taking into account the symmetries of the Weyl tensor,
we arrive at the form
\be
C_{acde;fg}{C_{b\ }}^{fge;cd}=-C_{acde;fg}{C_{b\ }}^{fcg;ed}
-C_{acde;fg}{C_{b\ }}^{fec;gd}=0,\label{CC22}
\ee
where the first term vanishes due to the symmetries of  the Weyl tensor and its derivatives and the second term due to lemma \ref{lemma_der}  and vanishing of the  rank-2, order-6 Weyl FKWC basis.

Above, we have proven vanishing of all rank-2 tensors quadratic in  $\nabla^{(k)} C$.  Due to this result and \eqref{commut_der}, covariant derivatives in a rank-2 tensor linear in $\nabla^{(k)} C$ commute. Vanishing of these linear terms is then a trivial consequence of the Bianchi identities and 
tracelessness of the Weyl tensor.

\end{proof}

\subsection{The Einstein case}

In the Einstein case, all rank-2 tensors constructed from the Weyl tensor without derivatives vanish due to \eqref{QGterm}.

Let us proceed with conserved rank-2 tensors quadratic in the Weyl tensor containing derivatives.
The FKWC basis \cite{FKWC} of rank-2, order-6 Weyl polynomials  
reduces again to \eqref{eqF}.

Differentiating \eqref{QGterm} twice, we obtain
\be
{C^{pqr}_{\ \ \ \ a;s}}^s C_{pqrb}+2C^{pqr}_{\ \ \ \ a;s} C_{pqrb}^{\ \ \ \ \ ;s}
+C^{pqr}_{\ \ \ \ a} {C_{pqrb;s}^{\ \ \ \ \ \ s}}=0. \label{derQGterm}
\ee
Using the Bianchi identities, \eqref{commut_der}, and the fact that all rank-2 tensors quadratic in the Weyl tensor vanish, 
we
find that the first and the last terms in \eqref{derQGterm} vanish.
Consequently, from \eqref{derQGterm}
\be
F_3=0. 
\ee

As in the Ricci flat case \ref{sec_III_Rflat}, we demand 
\be 
F_2=0.\label{vanF2}
\ee
Then from \eqref{F1F2}, $F_1=0$.

As in the Ricci flat case, for spacetimes satisfying $F_2=0$,  
the FKWC basis of rank-2, order-6 vanishes and thus do also all rank-2, 
order-6 Weyl polynomials. 

Using \eqref{commut_der} and vanishing of the FKWC basis, it follows that covariant derivatives in a rank-2 tensor of the form 
$\nabla^{(2)}C\nabla^{(2)}C$, $\nabla^{(3)}C\nabla^{(1)}C$,
and   $C\nabla^{(4)}C $  commute. If there is a summation within
one term $\nabla^{(2)}C$ in $\nabla^{(2)}C\nabla^{(2)}C$  
or in one term in  $\nabla^{(3)}C\nabla^{(1)}C$, or in  $C\nabla^{(4)}C $ then the resulting rank-2 tensor vanishes due
to the Bianchi identities, tracelessness of the Weyl tensor, and commuting of  covariant derivatives.
Then vanishing of all rank-2 order-6 tensors that we write symbolically as
$C^{(1)}C^{(2)}=0$ implies
\be 
C^{(1)}_{;f}C^{(2);f}=0.\label{C1fC2f}
\ee 
Hence, $\nabla^{(3)}C\nabla^{(1)}C$
and   $C\nabla^{(4)}C $  vanish and
the only rank-2 possibly non-vanishing tensor of the form $\nabla^{(2)}C\nabla^{(2)}C$ is
\eqref{CC22} that still vanishes using the same arguments as given for \eqref{CC22}.

Thus, we have proven

\begin{lem}
For type III, Einstein Kundt spacetimes, obeying $F_2=0$, all rank-2 
tensors of the form $\nabla^{(k)}C\nabla^{(l)}C$, $k+l\leq 4$ vanish.
\end{lem}

Let us prove using mathematical induction

\begin{prop}\label{prop_kvadrat}
	For type III, Einstein Kundt spacetimes, obeying $F_2=0$, all rank-2 
	tensors of the form $\nabla^{(k)}C\nabla^{(l)}C$  vanish.
\end{prop}

 We start by assuming that all rank-2 
tensors of the form $\nabla^{(k)}C\nabla^{(l)}C$, $k+l\leq p$  vanish.
Then
\begin{lem}\label{lem_commute}
 If all rank-2 
 tensors of the form $\nabla^{(k)}C\nabla^{(l)}C$, $k+l\leq p$,  vanish then the covariant derivatives in rank-2 tensors of the form  
 $\nabla^{(r)}C\nabla^{(s)}C$,
$r+s\leq p+2$,
 commute.
\end{lem}
\begin{proof}
 When commuting derivatives using \eqref{commut_der}, the additional terms  are 
 rank-2 tensors of the form $\nabla^{(r)}C\nabla^{(s)}C$, $r+s\leq p$
 that vanish by our assumption.
 \end{proof}
Then obviously,
\begin{lem}
	\label{lem_sumLambda}
	If all rank-2 
	tensors of the form $\nabla^{(k)}C\nabla^{(l)}C$, $k+l\leq p$,  vanish then 
	rank-2 tensors of the form 
	$\nabla^{(r)}C\nabla^{(s)}C$, $r+s\leq p+2$, 
vanish if there is a summation within one term.  
\end{lem}
\begin{proof}
	We commute the repeated dummy indices to the first position and then employ the Bianchi
	identities and the tracelessness of the Weyl tensor.  
\end{proof}

This further implies,
\begin{lem}\label{lem_C1C2lambda}
	If all rank-2 tensors of the form $C^{(1)}C^{(2)}=\nabla^{(k)}C\nabla^{(l)}C$, $k+l\leq p$ vanish then also
	 $C^{(1)}_{;e}C^{(2);e}=0$.
\end{lem}
\begin{proof}
	This can be shown by differentiating $C^{(1)}C^{(2)}=0$ twice and using
	lemma \ref{lem_sumLambda}.  
\end{proof}

\begin{proof}
	Now let us prove proposition \ref{prop_kvadrat}.
	
	We have assumed that all rank-2 
	tensors of the form $\nabla^{(k)}C\nabla^{(l)}C$, $k+l\leq p$, 
	vanish. We want to show that then also all rank-2 
	tensors of the form 
	$\nabla^{(r)}C\nabla^{(s)}C$, $r+s\leq p+2$,
	vanish.

	Using lemma \ref{lem_commute} and the Bianchi identities, without loss of generality, all case reduce to the following two cases 
	\beq
	& &
	C_{a. b.;\underbrace{....}_{r}}C_{....;\underbrace{....}_{s}}\,,\\
	& &	C_{a. ..;\underbrace{....}_{r}}C_{b...;\underbrace{....}_{s}}\,.
	\eeq
		If there is a summation within one term then by lemma
	\ref{lem_sumLambda}, the rank-2 tensor vanishes. Otherwise,
	$r=s+2$ or $r=s$, respectively. Then by lemma \ref{lem_C1C2lambda}, it reduces to (non-)vanishing of
	$C_{a.b.;..}C^{....}$ and 	\eqref{CC22}, respectively, 
	which was discussed earlier.

\end{proof}

The discussion of rank-2 tensors linear in $\nabla^{(k)}C$
is straightforward. The derivatives in  $\nabla^{(2)}C$
commute due to eqs. \eqref{QGterm} and \eqref{commut_der}.
Then all such rank-2 tensors vanish due to Bianchi identities
and tracelessness of the Weyl tensor. If all rank-2
tensors linear in $\nabla^{(k)}C$ vanish then using 
\eqref{commut_der} and proposition \ref{prop_kvadrat}, 
 all rank-2
tensors linear in $\nabla^{(k+2)}C$   vanish as well.
Thus, by mathematical induction all rank-2
tensors linear in $\nabla^{(p)}C$ for arbitrary $p$ vanish.

{ This  concludes the proof of proposition \ref{propIII}.}

\subsection{A type III  non-recurrent universal metric }
\label{sec_F2van}

Let us present an explicit example of a type III universal spacetime
with $\tau\not= 0$. In this section we use the real null basis
and corresponding formalism (see e.g. \cite{OrtPraPra12rev}).

In this case, the necessary condition for universality $F_2=0$ (see proposition \ref{propIII}) 
 reads
\beq
F_2\!\!\!&=&\!\!\!48\ell_a\ell_b 
\Psi'_i\tau_j (2\Psi'_j\tau_i-\Psi'_i\tau_j)\\
\!\!\!&=&\!\!\!
48\ell_a\ell_b[\tau_2 (\Psi'_3+ \Psi'_2)+\tau_3 (\Psi'_3- \Psi'_2)]
[\tau_2 (\Psi'_2- \Psi'_3)+\tau_3 (\Psi'_2+ \Psi'_3)]=0,
\nonumber
\eeq
hence
\be
\tau_2 (\Psi'_3\pm \Psi'_2)=\tau_3 (\Psi'_2\mp \Psi'_3).\label{condF2}
\ee

Type III
 Ricci-flat Kundt spacetimes  with $\tau\not=0$ admit a metric
 \cite{Stephanibook}
\be
\d s^2= -2\d u (\d r+ W_2\d x - W_3\d y +H\d u)+\d x^2 + \d y^2,
\ee
where
\beq
W_2&=&-\frac{2r}{x}+W^0_2 (u,x,y),\\
W_3&=&W^0_3 (u,x,y),\\
H&=&-\frac{r^2}{2x^2}+r \left(\frac{W^0_2}{x}+h_1(u)\right)+ H^0(u,x,y),
\eeq
where
\beq
W_2^0,_x&=&W_3^0,_y,\nonumber\\
W_2^0,_y&=&-W_3^0,_x \label{holom}
\eeq
(in the complex notation, the function $W_2^0+{\rm i} W_3^0$ is holomorphic)
and $H^0$ is subject to an additional b.w. $-2$ Einstein equation
\cite{Stephanibook}.

In the adapted null frame
\beq
\bl&=&\d u,\\
\bn&=& - (\d r+ W_2\d x - W_3\d y +H\d u),\\
\bm^{(2)}&=&\d x,\\
\bm^{(3)}&=&\d y,
\eeq
we obtain 
\beq
\tau_2\!\!\!&=&\!\!\!-1/x, \ \ \ \ \ \ \ \ \ \ 
\tau_3=0,\\
\Psi_2'\!\!\!&=&\!\!\!-\frac{1}{2x}W^0_2,_x,\ \ \ 
\Psi_3'=-\frac{1}{2x}W^0_2,_y.\label{III_Psi3}
\eeq

The condition $F_2=0$ 
 \eqref{condF2} implies
\be
\Psi'_3=\mp\Psi'_2,\label{condF2j}
\ee
which gives
\beq
W_2^0(u,x,y)=g(x\pm y)+f_2(u),\\
W_3^0(u,x,y)=g(y\mp x)+f_3(u).
\eeq
By \eqref{holom}, this reduces to
\beq
W_2^0(u,x,y)&=&F(u)(x\pm y)+c_2(u),\\
W_3^0(u,x,y)&=&F(u) (y\mp x)+c_3(u).
\eeq

\section{Conclusions}
\label{concl}

In four dimensions, we have obtained stronger results on universal spacetimes than in  previous works in arbitrary  dimensions \cite{HerPraPra14,univII}. 

In four dimensions, we have proved that universal spacetimes are necessarily algebraically special and Kundt.  Furthermore, in addition to the necessary and sufficient conditions for universality for type  N already known in arbitrary dimension, we have found necessary  and sufficient conditions for type III.  We have pointed out that apart from type III spacetimes with a  recurrent null vector,
the non-recurrent case is also universal provided $F_2$ (as defined in proposition \ref{propIII}) vanishes.

For type D, the universality condition is very restrictive, allowing only for direct products of two 2-spaces of  constant and equal curvatures. Type II universal spacetimes then reduce to these type D backgrounds in an appropriate limit. In contrast to types III and N, type II and D universal spacetimes necessarily  admit  recurrent null vector.

{{In table \ref{table1}, known necessary/sufficient conditions for universality
for various algebraic types in four and higher dimensions are summarized. }}

\begin{table}[h]
	\begin{center}
		\begin{tabular}{|c||l|l|} 
			\hline
			type & 4D & HD
			\\[0.5mm]
			\hline\hline
			I/G  & $\not\exists$ (prop. \ref{propalgspec})& \\[1mm]\hline
			II  & N: & \\[1mm] 
		    &	$\bullet$  E+K+ $(\tau=0)\ (\Rightarrow\Lambda\not= 0)$ (prop. \ref{lem_K}) 
			
			& $\bullet$ $\not\exists $ 5D (theorem 1.2 \cite{univII})\\[1mm]
			&$\bullet$ additional conditions (e.g. from eq. \eqref{GT-RRnablaC})  & $\bullet$ $\tau=0$ is not necessary \\[1mm] 
			& $ \bullet $ 
			  extensions of univ. type D 
			(sec. \ref{sec_seedII})
			& $ \bullet $  $\tau=0\ \Rightarrow\ \Lambda\not=0$ (prop. 5.1 \cite{univII})\\[1mm]
			 & &
			$ \bullet $ S: universal Kundt extensions  \\[1mm] 
			 & & of type D univ. spacetimes (prop. 6.2 \cite{univII})\\[1mm]\hline
			 	D  & NS: direct product of 2 2-spaces 
			 & S: direct product of $N$ max. sym. $n$-spaces \\[1mm]
			 & with the same Ricci scalar 
			   (prop. \ref{propD}) & with the  same  Ricci scalar 
			   (prop. 6.1 \cite{univII})\\[1mm] \hline  
			III    & NS: E+K + $(F_2 =0)$ 
			 (prop. \ref{propIII}) & 
			S: E+K + $({C_a}^{cde}C_{bcde}=0) $
			+ $(\tau=0)$ \\ & & (theorem 1.4. \cite{HerPraPra14})
			\\[1mm]\hline 
			N  & NS: E+K (theorem 1.3. \cite{HerPraPra14}) & NS: E+K (theorem 1.3. \cite{HerPraPra14})  \\[1mm]
			\hline
		\end{tabular} \\[2mm]
		\caption{Universal spacetimes  in four and higher dimensions (HD), known necessary (N)/sufficient (S) conditions for various algebraic types. All universal spacetimes are Einstein (E) and CSI (theorem 1.2 \cite{HerPraPra14}) and in four dimensions, they are all necessarily Kundt (K) (proposition \ref{propalgspec}). }
		\label{table1}
		\label{tbl:01}
	\end{center}
\end{table} 

Let us conclude with a discussion of universality for VSI spacetimes (spacetimes with all scalar curvature invariants vanishing \cite{Pravdaetal02}).
Although all curvature invariants
in VSI spacetimes vanish, conserved
rank-2 tensors may be non-vanishing
(in contrast to what seems to be suggested in \cite{ColeyLett}).
For example, as noted 
in \cite{HerPraPra14}, in higher dimensions $C_{acde} C_{b}^{\ cde}$ is in general non-vanishing  for type III VSI spacetimes
and $F_2$ is in general non-vanishing for type III VSI spacetimes with $\tau_i \not=0$ even in four dimensions.
Thus, although many VSI spacetimes are universal and thus represent  an interesting class of spacetimes in this context,  VSI is neither a sufficient, nor a necessary condition for universality.

\section*{Acknowledgements}
AP and VP would like to thank University of Stavanger for its hospitality while part of this work was carried out.
This work was supported from the research plan RVO: 67985840, the research grant GA\v CR 13-10042S (VP, AP) and  through the Research Council of Norway, Toppforsk
grant no. 250367: \emph{Pseudo-Riemannian Geometry and Polynomial Curvature Invariants:
Classification, Characterisation and Applications} (SH).

\providecommand{\href}[2]{#2}

\end{document}